\documentclass[12pt,pra,reprint,superscriptaddress,amsmath,amsfonts,amssymb]{revtex4-1}

\usepackage[pretty]{revquantum}
\usepackage[english]{babel}
\usepackage{graphicx}
\usepackage[caption=false]{subfig}
\usepackage{mathrsfs}
\usepackage{amsmath}
\usepackage{amsthm}
\usepackage{thmtools, thm-restate}
\usepackage{bm}
\usepackage{xcolor}
\usepackage{hyperref}
\usepackage{hypernat}
\usepackage[utf8]{inputenc}
\usepackage[T1]{fontenc}
\usepackage{cleveref}
\usepackage{soul}

\declaretheorem[parent=section,name=Theorem]{thm}
\declaretheorem[sibling=theorem,name=Lemma]{lem}
\declaretheorem[style=definition,sibling=thm]{definition}
\declaretheorem[style=remark,sibling=thm]{example}

\declaretheorem[sibling=thm]{proposition}
\declaretheorem[sibling=thm]{corollary}

\begin{document}
\title{Noncontextuality Inequalities from Antidistinguishability}

\author{Matthew Leifer}
\affiliation{Schmid College of Science and Technology, Chapman University, One University Dr., Orange, CA 92866, USA}
\affiliation{Institute for Quantum Studies, Chapman University, One University Dr., Orange, CA 92866, USA}
\author{Cristhiano Duarte}
\affiliation{Schmid College of Science and Technology, Chapman University, One University Dr., Orange, CA 92866, USA}
\email[Corresponding author: ]{crsilva@chapman.edu}

\date{\today}
\begin{abstract}
Noncontextuality inequalities are usually derived from the distinguishability properties of quantum states, i.e.\ their orthogonality. Here, we show that \emph{anti}distinguishability can also be used to derive noncontextuality inequalities.  The Yu-Oh 13 ray noncontextuality inequality  can be re-derived and generalized as an instance of our antidistinguishability method. For some sets of states, the antidistinguishability method gives tighter bounds on noncontextual models than just considering orthogonality, and the Hadamard states provide an example of this. We also derive noncontextuality inequalities based on mutually unbiased bases and symmetric informationally complete POVMs.  Antidistinguishability based inequalities were initially discovered as overlap bounds for the reality of the quantum state. Our main contribution here is to show that they are also noncontextuality inequalities.
\end{abstract}

\maketitle

\section{Introduction}

\label{Intro}

Quantum contextuality has its origins in work of Bell \cite{Bell_ProblemHiddenVariables_1966}, and Kochen and Specker \cite{Kochen_ProblemHiddenVariables_1967},
where they proved a no-go theorem ruling out deterministic hidden variable theories in which the value assigned to an observable is 
independent of how you measure it.  In recent years, contextuality has attracted increasing attention for its role in quantum information 
processing advantages \cite{Spekkens_2009, Kleinmann_2011, Grudka_2014, Chailloux_2016, Abramsky_2017, Schmid_2018, Duarte_ResourceTheoryContextuality_2018, Ghorai_OptimalQuantumPreparation_2018} and explaining the power of quantum computation \cite{Galvao_2005, Cormick_2006, Anders_2009, Howard_ContextualitySuppliesMagic_2014, Hoban_2014, Karanjai_2018, Abramsky_2017, Frembs_ContextualityResourceMeasurementbased_2018, Raussendorf_2019, Catani_2019}. 
For these purposes, it is useful to find new classes of noncontextuality inequalities and to find the tightest possible bounds on them.

Noncontextuality inequalities are usually based on the orthogonality properties of sets of quantum states, or, equivalently, they are based on our ability to perfectly distinguish sets of quantum states.  A powerful method for deriving bounds on noncontextuality inequalities from the orthogonality graphs of events has been developed by Cabello, Severini and Winter (CSW) \cite{Cabello_NonContextualityPhysical_2010, Cabello_GraphTheoreticApproachQuantum_2014}.  A similar method, also exploring our ability of perfectly distinguish between objects, has been applied to Bell inequalities to provide tighter bounds \cite{RDLTC14}.

In this paper, we show that the \emph{anti}distinguishability properties \cite{Leifer_QuantumStateReal_2014} \footnote{Antidistinguishability also goes by the names \emph{PP-incompatibility} \cite{Caves_ConditionsCompatibilityQuantumstate_2002} and \emph{conclusive exclusion of quantum states}  \cite{Bandyopadhyay_ConclusiveExclusionQuantum_2014}.} of quantum states can also be used to derive noncontextuality inequalities.  Our method reproduces the inequality used in the Yu-Oh 13 ray proof of contextuality \cite{Yu_StateIndependentProofKochenSpecker_2012}, giving more intuition behind its structure and allowing us to propose several generalizations.  In some cases, when we apply both the CSW method and our method to the same set of states, we get a much tighter bound on the noncontextuality inequality.

The concept of antidistinguishability  was first proposed in \cite{Caves_ConditionsCompatibilityQuantumstate_2002}, and played a key role in the proof of the Pusey, Barrett and Rudolph (PBR) theorem \cite{Pusey_RealityQuantumState_2012}.  The aim of the PBR theorem was to address the question of whether the quantum state is a state of reality, akin to a point in phase space for a classical particle (known as the $\psi$-ontic view of quantum states), or a state of knowledge, more akin to a probability distribution over phase space (known as the $\psi$-epistemic view).  The $\psi$-epistemic view has a lot of advantages, as many otherwise puzzling phenomena, including the indistinguishability of non-orthogonal quantum states and the no-cloning theorem, are easily explained by the fact that the probability distributions representing non-orthogonal quantum states can overlap in a $\psi$-epistemic model \cite{Spekkens_EvidenceEpistemicView_2007, Leifer_QuantumStateReal_2014, Jennings_2015}.  The PBR theorem showed that, within a standard framework for realist models, known as the ontological models framework \cite{Harrigan_EinsteinIncompletenessEpistemic_2010}, only $\psi$-ontic models are possible.  

However, the PBR theorem is based on additional assumptions beyond the bare ontological models framework, and these assumptions have attracted criticism \cite{Hall_GeneralisationsRecentPuseyBarrettRudolph_2011,Emerson_WholeGreaterSum_2013,Schlosshauer_ImplicationsPuseyBarrettRudolphQuantum_2012}.  Subsequently, there was an effort to determine what could be proved about the reality of the quantum state without such additional assumptions.  It was shown that $\psi$-epistemic models exist in all finite Hilbert space dimensions \cite{Lewis_2012, Aaronson_PsiEpistemicTheoriesRole_2013}.  This led to the definition of \emph{maximally} $\psi$-epistemic models \cite{Leifer_MaximallyEpistemicInterpretations_2013,Maroney_HowStatisticalAre_2012,Harrigan_OntologicalModelsInterpretation_2007} \footnote{Non maximally $\psi$-epistemic models were originally defined under the name \emph{deficient} models in \cite{Harrigan_OntologicalModelsInterpretation_2007}.} and the study of overlap bounds for probability distributions in ontological models \cite{Barrett_2014, Leifer_PsEpistemicModelsAre_2014, Branciard_HowPsepistemicModels_2014, Ringbauer_2015, Knee_2017}. 

 In order for the $\psi$-epistemic explanations of quantum phenomena to work, it is not enough that there is just some amount of overlap of probability distributions, but the overlap should be comparable to the degree of indistinguishability of the quantum states.  This was ruled out by showing that it would imply that the ontological model is noncontextual \cite{Leifer_MaximallyEpistemicInterpretations_2013, Leifer_QuantumStateReal_2014}, which is ruled out by existing contextuality proofs.  Noncontextuality inequalities can then be used to bound the degree of overlap in an ontological model, and one class of overlap bounds is based on doing exactly this with CSW inequalities \cite{Leifer_PsEpistemicModelsAre_2014}.  
 
 However, another class of overlap inequalities was proposed in the literature based on the antidistinguishability of quantum states \cite{Barrett_2014, Branciard_HowPsepistemicModels_2014, Ringbauer_2015, Knee_2017} and it was not obvious whether these have anything to do with contextuality.  Our main result is to re-derive these inequalities as noncontextuality inequalities, which means that all the antidistinguishability overlap bounds in the literature can now be reinterpreted as noncontextuality inequalities.  We also re-derive and generalize some other noncontextuality inequalities that have appeared in the literature \cite{Yu_StateIndependentProofKochenSpecker_2012, Bengtsson_2012} by showing that they are examples of the antidistinguishability-based construction.

The rest of this paper is organized as follows.  In \S\ref{ConSce} we review the mathematical framework of \emph{contextuality scenarios} as developed in \cite{Acin_2015}, slightly generalized to allow for both measurements with a fully specified set of outcomes and those with an under-specified set.  This is the framework in which we prove our results.  In \S\ref{Anti}, we give a definition of antidistinguishability for contextuality scenarios that generalizes the existing definition for quantum states.  \S\ref{Ineq} contains our main results.  It introduces the notions of \emph{strong and weak pairwise antisets}, which are sets of outcomes in a contextuality scenario such that any pair of them together with another outcome in a specified set is antidistinguishable.  Our main result shows that there is a noncontextuality inequality associated with any pairwise antiset.  \S\ref{Exa} gives examples of pairwise antisets in quantum theory and their associated noncontextuality inequalities, showing how existing inequalities can be re-derived and generalized in this approach.  The proof of our main results is given in  \S\ref{Proof} and \S\ref{Conc} concludes with a summary and outlook.  

\section{Contextuality Scenarios}
\label{ConSce}

This section reviews a slightly generalized version of the contextuality scenario framework developed in \cite{Acin_2015}.  After introducing the basic definitions, we review the concepts of \emph{value functions} in \S\ref{Cont:Value} and \emph{quantum models} in \S\ref{Cont:Quant}.  These describe the possible noncontextual and quantum realizations of contextuality scenarios respectively  \S\ref{Cont:States} reviews the concept of \emph{states} on a contextuality scenario, which describe the observable probabilities in noncontextual, quantum, and more general models.  The aim is to arrive at a general framework for discussing \emph{noncontextuality inequalities}, which are inequalities satisfied by noncontextual states, but not necessarily quantum or more general states.

\begin{definition}
    A \emph{contextuality scenario} $\mathfrak{C}$ is a structure $\mathfrak{C} = (X,\mathcal{M},\mathcal{N})$ where
    \begin{itemize}
        \item $X$ is a set of \emph{outcomes}.
        \item $\mathcal{M}$ is a set of subsets of $X$ such that if $M,M' \in \mathcal{M}$ then $M'\not\subset M$.  An $M \in \mathcal{M}$ is called a \emph{(measurement) context}.
        \item $\mathcal{N}$ is a set of subsets of $X$ such that if $M \in \mathcal{M}$ then $M \not\in \mathcal{N}$ and if $N,N'\in \mathcal{N}$ then $N' \not\subset N$. An $N \in \mathcal{N}$ is called a \emph{maximal partial (measurement) context}.
    \end{itemize}
    
Finally, a contextuality scenario is \emph{finite} if  $X$ is a finite set.
\end{definition}

The idea of a contextuality scenario is that you have a system on which you can perform several different measurements.  $X$ is the set of all possible measurement outcomes.  A context $M\in \mathcal{M}$ is the full set of distinct outcomes that can occur in some possible measurement.  Note that the condition that $\mathcal{M}$ contains no sets that are subsets of other sets in $\mathcal{M}$ is not usually imposed in the literature, but is true of all the interesting examples.

A maximal partial context $N \in \mathcal{N}$ is a set of outcomes that can occur as the outcome of some possible measurement, but not necessarily the full set.  We allow for the set of outcomes of some measurements to be incompletely specified.  For example, a failure to detect the system at all could count as an unspecified outcome.  In this respect, our definition of a contextuality scenario is slightly more general than that of \cite{Acin_2015}, which only has $\mathcal{M}$.  

Note that all the contextuality scenarios we use in this paper are finite, so we will assume this going forward without further comment.

A contextuality scenario with no maximal partial contexts is a specific type of hypergraph, and, in general, a contextuality scenario can be seen as is a generalization of a hypergraph with two kinds of hyperedges \footnote{Instead of generalizing the concept of a hypergraph, we could simply have considered a coloring process on the hyperedges of such a hypergraph. Assigning different colours to different kinds of hyperedges, we would end up drawing essentially the same graphs as shown in fig.\ref{fig:classical}}.  We can draw diagrams of them by denoting contexts with solid lines and maximal partial contexts with dashed lines, as in the following examples.

\begin{example}
    \label{exa:class}
    A \emph{classical} contextuality scenario has a finite set $X$ of outcomes, $\mathcal{M} = \{X\}$, and $\mathcal{N} = \emptyset$.  A \emph{partial classical} contextuality scenario has a finite set $X$ of outcomes, $\mathcal{M} = \emptyset$, and $\mathcal{N} = \{X\}$.  In words, every set of outcomes can, and indeed does, occur together in a single realization of a measurement.  These scenarios are depicted in \cref{fig:classical}
\end{example}

\begin{figure}[!htb]
    \centering
    \subfloat[A classical contextuality scenario.]{\includegraphics[width=.8\linewidth]{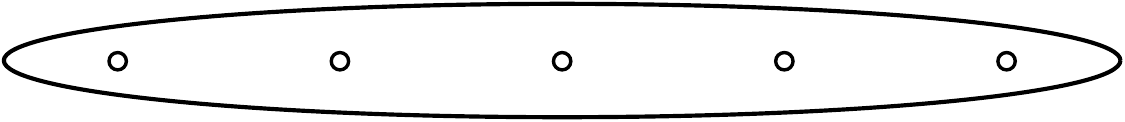}}
    
    \subfloat[A partial classical contextuality scenario.]{\includegraphics[width=.8\linewidth]{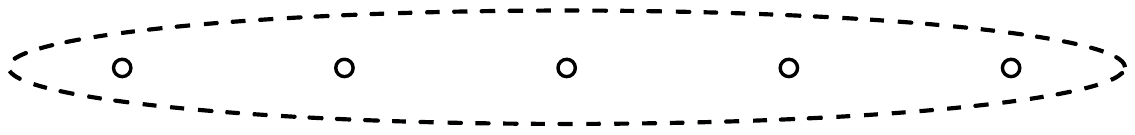}}
    \caption{Examples of classical contextuality scenarios with 5 outcomes.}
    \label{fig:classical}
\end{figure}

\begin{example}
    \label{exa:Speck}
    The \emph{Specker Triangle} \cite{Specker_1960} is the contextuality scenario with $X = \{a,b,c\}$, $\mathcal{M} = \{\{a,b\},\{b,c\},\{c,a\}\}$, and $\mathcal{N} = \emptyset$, as shown in \cref{fig:specker}.
\end{example}

\begin{figure}[!htb]
    \centering
    \includegraphics[scale=0.6]{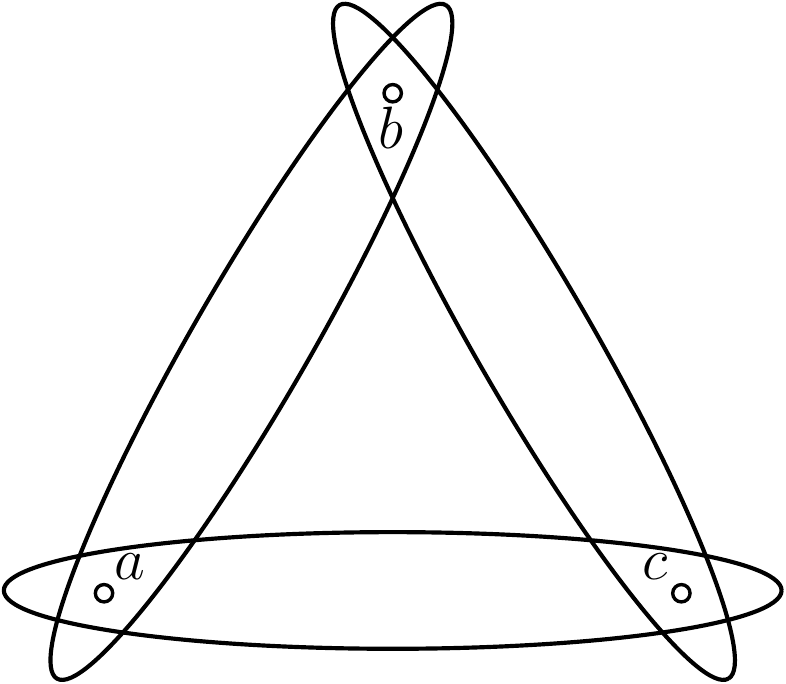}
    \caption{The Specker Triangle}
    \label{fig:specker}
\end{figure}

\begin{example}
    \label{exa:anti}
    The following is an example of an \emph{antidistinguishability} scenario that we will make use of later.  It has both contexts and maximal partial contexts. Set $X = \{a_1,a_2,a_3,a^{\perp}_1,a^{\perp}_2,a^{\perp}_3\}$, $\mathcal{M} = \{\{a^{\perp}_1,a^{\perp}_2,a^{\perp}_3\}\}$, and $\mathcal{N} = \{\{a_1,a^{\perp}_1\},\{a_2,a^{\perp}_2\},\{a_3,a^{\perp}_3\}\}$.  This is shown in \cref{fig:antidistinguish}
\end{example}

\begin{figure}[!htb]
    \centering
    \includegraphics[scale=0.6]{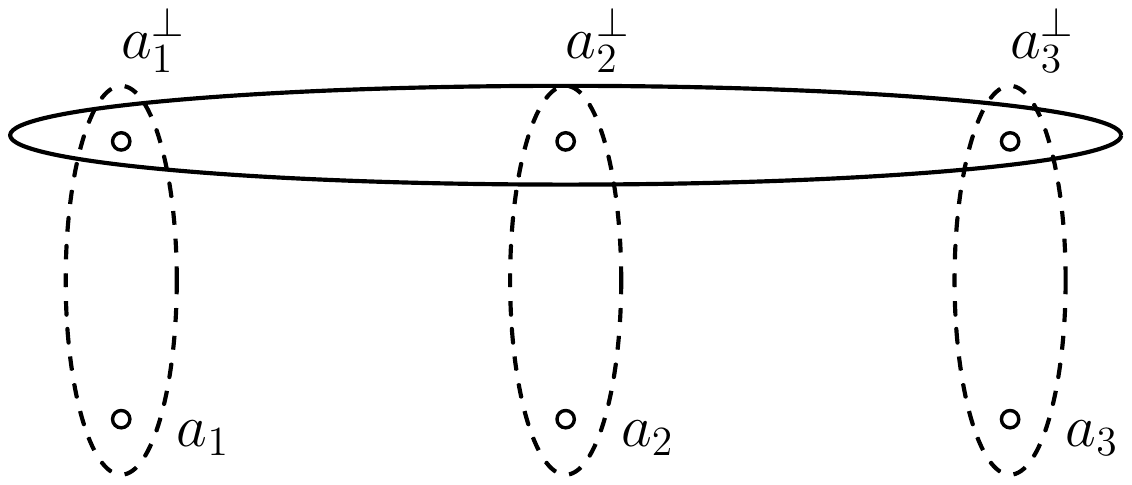}
    \caption{An antidistinguishability scenario}
    \label{fig:antidistinguish}
\end{figure}

\begin{example}
    \label{exa:qanti}
    A \emph{quantum} contextuality scenario is constructed as follows.  Let $X$ be a set of pure states (unit vectors with vectors differing by a global phase identified) in a Hilbert space $\mathcal{H}$.  A subset $M \subseteq X$ is in $\mathcal{M}$ iff $M$ is an orthonormal basis.  A subset $N \subseteq X$ is in $\mathcal{N}$ iff the states it contains are pairwise orthogonal, it is not a basis (i.e.\ it is incomplete), and it is not a subset of any other $M \in \mathcal{M}$ or $N \in \mathcal{N}$.
    
    As an example, consider the six states
    \begin{align}
    \ket{a_1} &= \begin{pmatrix}
           1 \\
           0 \\
           0
         \end{pmatrix},
         \ket{a_2} = \frac{1}{\sqrt{3}} \begin{pmatrix}
           1 \\
           1 \\
           1
         \end{pmatrix},
         \ket{a_3} = \frac{1}{\sqrt{3}} \begin{pmatrix}
           -1 \\
           1 \\
           1
         \end{pmatrix}
         \label{Eq:ExampleAnti} \\
    \ket{a_1^\perp} &= \begin{pmatrix}
           0 \\
           1 \\
           0
         \end{pmatrix},
         \ket{a_2^\perp} = \frac{1}{\sqrt{2}} \begin{pmatrix}
           1 \\
           0 \\
           -1
         \end{pmatrix},
         \ket{a_3^\perp} = \frac{1}{\sqrt{2}} \begin{pmatrix}
           1 \\
           0 \\
           1
         \end{pmatrix}
      \end{align}
      
      Inspection of the orthogonality relations shows that the quantum contextuality scenario generated by these states is the antidistinguishability scenario of \cref{exa:anti}.
\end{example}

\subsection{Value Functions}

\label{Cont:Value}

\begin{definition}
    A \emph{value function} $v:X \rightarrow \{0,1\}$ on a contextuality scenario $\mathfrak{C} = (X,\mathcal{M},\mathcal{N})$ is a function that assigns a value $0$ or $1$ to every outcome such that
    \begin{itemize}
        \item For every $M \in \mathcal{M}$, $v(a) = 1$ for exactly one $a\in M$.
        \item For every $N \in \mathcal{N}$, $v(a) = 1$ for at most one $a \in N$.
    \end{itemize}
    
    The set of all value functions on $\mathfrak{C}$ is denoted $V_{\mathfrak{C}}$.
\end{definition}

\begin{definition}
    For a contextuality scenario $\mathfrak{C} = (X,\mathcal{M},\mathcal{N})$ and an outcome $a \in X$, an $a$-\emph{definite} value function is a value function such that $v(a) = 1$.  The set of $a$-definite value functions is denoted $V_a$.
\end{definition}

The idea of a value function is that it is a deterministic assignment of outcomes to every measurement.  For every context, one of the outcomes must occur because the context contains the full set of possible outcomes of that measurement, so the chosen outcome is assigned the value $1$.  For partial contexts, one of the unspecified outcomes may be the actual outcome of the measurement, so we only demand that at most one outcome is assigned the value $1$.

Value functions are noncontextual because they are defined directly on $X$.  A given $a\in X$ may occur in more than one (maximal partial) context, as in the Specker triangle, but the value assigned to the outcome is not allowed to depend on which context is being measured.

Note that not all contextuality scenarios have value functions.  For example, in the Specker triangle, we would have to assign value $1$ to exactly one of each pair $\{a,b\}$, $\{b,c\}$ and $\{a,c\}$.  By symmetry, we can start by assigning $1$ to any of the three outcomes, so let's choose $a$.  Then we must assign $0$ to $b$ because of the pair $\{a,b\}$ and $0$ to $c$ because of the pair $\{a,c\}$.  But then neither $b$ nor $c$ is assigned the value $1$, which contradicts the requirement that exactly one of the pair $\{b,c\}$ is assigned the value $1$.  

There are also quantum contextuality scenarios that have no value functions.  This is the content of the Bell-Kochen-Specker theorem \cite{Bell_ProblemHiddenVariables_1966, Kochen_ProblemHiddenVariables_1967}.

\subsection{Quantum Models}

\label{Cont:Quant}

\begin{definition}
    A \emph{quantum model} of a contextuality scenario $\mathfrak{C} = (X,\mathcal{M},\mathcal{N})$ consists of
    \begin{itemize}
        \item A choice of Hilbert space $\mathcal{H}$.
        \item For every $a\in X$, a projection operator $P_a$ onto a closed subspace of $\mathcal{H}$ such that:
        \begin{itemize}
            \item For every $M \in \mathcal{M}$, $\sum_{a \in M} P_a=I$, where $I$ is the identity operator.
            \item For every $N \in \mathcal{N}$, $a,b\in N$ and $a\neq b$, $P_aP_b=0$.
        \end{itemize}
    \end{itemize}
\end{definition}

A quantum model represents every context by a projective quantum measurement and every maximal partial context by a subset of the projectors in such a measurement. 

Not all contextuality scenarios have a quantum model.  The Specker triangle is again an example.  The context $\{a,b\}$ implies that $P_a + P_b = I$, so $P_b = I - P_a$, and $\{a,c\}$ that $P_c = I-P_a$.  Then, $\{b,c\}$ implies that $P_b+P_c = I$, and substituting the previous two equations into this gives $P_a = I/2$, which is not a projection operator. 

Clearly, if we start with a quantum contextuality scenario then it has a quantum model, i.e.\ the projectors onto the states that define the model, but it also has other quantum models.  For example, applying a unitary transformation to all the states preserves their orthogonality structure, so it gives us another quantum model.

The Bell-Kochen-Specker theorem implies that there are contextuality scenarios that have a quantum model, but no value functions.  However, whenever there is a value function there is a quantum model.

\begin{proposition}
    \label{prop:classinquant}
    If a contextuality scenario $\mathfrak{C}=(X,\mathcal{M},\mathcal{N})$ has a value function then it also has a quantum model.
\end{proposition}
\begin{proof}
    Let $\mathcal{H} = \mathcal{H}_{V_{\mathfrak{C}}}$, i.e.\ the Hilbert space with orthonormal basis vectors labeled by the elements of $V_{\mathfrak{C}}$.  For every $a\in X$, define the projector
    \[P_a = \sum_{v\in V_a} \ket{v}\bra{v}. \]
    This defines a quantum model.  
    
    To see this, let $M \in \mathcal{M}$.  Notice that the sets $V_a$ for $a \in M$ are disjoint because each value function assigns value $1$ to exactly one element of $M$.  They also cover the whole set $V_{\mathfrak{C}}$ because every value function assigns value $1$ to some element of $M$. Thus,
    \begin{align*}
        \sum_{a \in M} P_a & = \sum_{a\in M} \sum_{v\in V_a} \ket{v}\bra{v} \\
        & = \sum_{v\in V_{\mathfrak{C}}} \ket{v}\bra{v} = I.
    \end{align*}
    
    Now let $N \in \mathcal{N}$ and consider $a,b\in N$, $a\neq b$.  We have
    \[
        P_aP_b = \sum_{v\in V_a}\sum_{w \in V_b} \ket{v}\braket{v | w}\bra{w} = 0, 
    \]
    because $V_a$ and $V_b$ are disjoint.
\end{proof}

\subsection{States}

\label{Cont:States}

\begin{definition}
    A \emph{state} $\omega: X \rightarrow [0,1]$ on a contextuality scenario $\mathfrak{C} = (X,\mathcal{M},\mathcal{N})$ is a function that assigns a probability to every outcome such that
    \begin{itemize}
        \item For all $M \in \mathcal{M}$,
        \[\sum_{a\in M} \omega(a) = 1.\]
        \item For all $N \in \mathcal{N}$,
        \[\sum_{a\in N} \omega(a) \leq 1.\]
    \end{itemize}
    
    The set of states on $\mathfrak{C}$ is denoted $S_{\mathfrak{C}}$.
\end{definition}

A state is an assignment of probabilities to outcomes that is compatible with every (maximal partial) context having a well-defined probability distribution.  For the maximal partial contexts, we only demand that the probabilities add up to something less than or equal to $1$ because it is possible to put probability weight on the unspecified outcomes.

For a classical scenario, the states are exactly the probability distributions on $X$ and for a partial classical scenario, they are the sub-normalized probability distributions on $X$.

The Specker triangle has exactly one state: $\omega(a) = \omega(b) = \omega(c) = \frac{1}{2}$, which can be obtained by solving the equations defining the state space.

There are also scenarios with no states, the simplest being $X = \{a_1,a_2,a_3,b_1,b_2,b_3\}$, $\mathcal{M} = \{\{a_1,a_2,a_3\},\{b_1,b_2,b_3\},\{a_1,b_1\},\{a_2,b_2\},\{a_3,b_3\}\}$ and $\mathcal{N} = \emptyset$.  The first two contexts require $\omega(a_1) + \omega(a_2) + \omega(a_3) = 1$ and $\omega(b_1) + \omega(b_2) + \omega(b_3) = 1$, so that
\[\sum_{j=1}^3 \left [ \omega(a_j) + \omega(b_j) \right ] = 2\].  However, the last three contexts require $\omega(a_j) + \omega(b_j) = 1$ for $j=1,2,3$, and hence
\[\sum_{j=1}^3 \left [ \omega(a_j) + \omega(b_j) \right ] = 3,\]
which is a contradiction.

We can represent a state by a vector in the space $\mathbb{R}^X$ where, for each $a \in X$, $\omega(a)$ is the component of the vector in the direction corresponding to $a$.  In this representation, the state space is a convex polytope because it is defined by a finite set of linear equations and inequalities and every component is bounded between $0$ and $1$.

\begin{definition}
    A \emph{Kochen-Specker (KS) noncontextual} state on a contextuality scenario $\mathfrak{C} = (X,\mathcal{M},\mathcal{N})$ is a state $\omega$ such that
    \[\omega(a) = \sum_{v\in V_{\mathfrak{C}}} p_v v(a),\]
    where $p_v$ is a probability distribution on $V_{\mathfrak{C}}$, i.e.\ $0 \leq p_v \leq 1$ and $\sum_{v \in V_{\mathfrak{C}}} p_v = 1$.
    
    The set of KS noncontextual states on $\mathfrak{C}$ is denoted $C_{\mathfrak{C}}$.  A state $\omega$ that is not contained in $C_{\mathfrak{C}}$ is called a \emph{contextual} state.
\end{definition}

Viewed as a subset of $\mathbb{R}^X$, $C_{\mathfrak{C}}$ is also a convex polytope because there are a finite number of value functions which define its vertices.

If we observe probabilities in an experiment that agree with a KS noncontextual state then we can imagine that there is always a definite noncontextual outcome for each measurement, and the observation of probabilities that differ from $0$ or $1$ is just due to our ignorance of which value function holds in each particular run of the experiment.  On the other hand, contextual states cannot be understood in this way.

\begin{definition}
    A \emph{quantum state} on a contextuality scenario $\mathfrak{C} = (X,\mathcal{M},\mathcal{N})$ is a state $\omega$ such that there exists a quantum model and a density operator $\rho$ on $\mathcal{H}$ (the Hilbert state of the model) for which
    \[\omega(a) = \Tr(P_a \rho).\]
    
    The set of quantum states on $\mathfrak{C}$ is denoted $Q_{\mathfrak{C}}$.
\end{definition}

The set of quantum states is the set of observable probabilities for a contextuality scenario that is realized as a quantum experiment.  If we find a contextual quantum state then this is a proof that quantum mechanics is contextual.  The set of quantum states is a compact convex set, but not necessarily a polytope \cite{Brunner13}.

\begin{definition}
    A \emph{state independent noncontextuality inequality} for a contextuality scenario $\mathfrak{C} = (X,\mathcal{M},\mathcal{N})$ is a linear inequality of the form
    \begin{equation}
        \label{eq:ineq}
        \sum_{a\in X} c_a \omega(a) \leq \gamma_c,
    \end{equation}    
    where $c_a,\gamma_c\in\mathbb{R}$, which is satisfied for all $\omega \in C_{\mathfrak{C}}$.
    
    A \emph{state dependent noncontextuality inequality} is an inequality of the form of \cref{eq:ineq} that is satisfied for all $\omega \in C_{\mathfrak{C}}$ that also satisfy some additional set of constraints.
\end{definition}

If, having derived a state independent noncontextuality inequality, we find a state $\omega$ such that $\sum_{a\in X}c_a \omega(a) > \gamma_c$, then this is a proof that $\omega$ is contextual.  The kind of additional constraints that might be imposed in a state dependent inequality are things like $\omega(a) = 0$ for some specified outcome.  In this case, if we find a state such that $\sum_{a\in X}c_a \omega(a) > \gamma_c$ that also satisfies the additional constraints, then this is a proof that $\omega$ is contextual.

Note, the inequalities that we derive in this paper have $c_a \in \{0,1\}$ for all $a \in X$, but more general inequalities are possible.

The terminology state independent/dependent \emph{inequality} that we have introduced here should be contrasted with the notions of state independent/dependent \emph{proofs} of contextuality, which are common in the literature \cite{Yu_StateIndependentProofKochenSpecker_2012}.  In a state independent proof, once a quantum model is fixed for a contextuality scenario, we find that $\sum_{a\in X} c_a \omega(a)$ is completely independent of the quantum state $\omega$ chosen so all quantum states are contextual in that model.  In a state dependent proof, the value varies with $\omega$, so whether the inequality is violated, and by how much it is violated, depends on the state chosen.  A state independent inequality can be the basis of either a state independent or dependent proof, depending on the details of the quantum model chosen, but a state dependent inequality necessarily leads to a state dependent proof, since the inequality does not hold for all choices of state.

\begin{example}[Klyachko Inequality \cite{Klyachko_2002, Klyachko_2008}]
    Consider the Klyachko contextuality scenario $\mathfrak{C} = (X,\mathcal{M},\mathcal{N})$ with $X = \{0,1,2,3,4\}$, $\mathcal{M} = \emptyset$ and $\mathcal{N} = \{\{0,1\},\{1,2\},\{2,3\},\{3,4\},\{4,0\}\}$ as depicted in \cref{fig:Klayatchko}.  Then,
    \[\sum_{a\in X}\omega(a) \leq 2,\]
    is a state independent noncontextuality inequality.
    
\begin{figure}[!htb]
    \centering
    \includegraphics[scale=0.6]{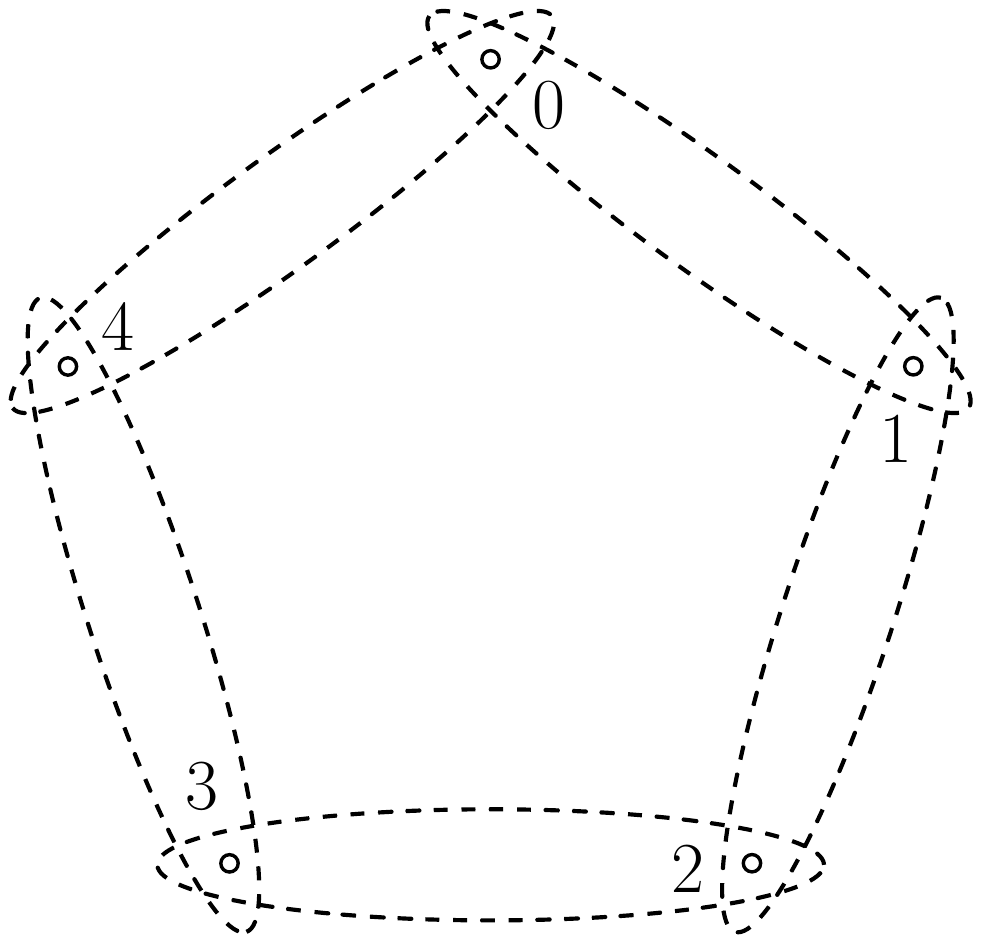}
    \caption{The Klyachko contextuality scenario}
    \label{fig:Klayatchko}
\end{figure}    
    
    To see this note that, for a KS noncontextual state of the form $\omega(a) = \sum_{v\in V_{\mathfrak{C}}} p_v v(a)$, we have
    \begin{align*}
        \sum_{a \in X} \omega_a & = \sum_{a\in X}\sum_{v\in V_{\mathfrak{C}}} p_v v(a) \\
        & = \sum_{v\in V_{\mathfrak{C}}} p_v \sum_{a\in X} v(a) \\
        & \leq \max_{v\in V_{\mathfrak{C}}} \left [ \sum_{a \in X} v(a) \right ],
    \end{align*}
    where the last line follows from convexity.
    
    It is easy to see that, for any $v \in V_{\mathfrak{C}}$, $v(0) + v(1) + v(2) + v(3) + v(4) \leq 2$.  By symmetry, we can start by assigning $v(0) = 1$, which implies that $v(1) = v(4) = 0$.  Then we could assign $v(2) = 1$, which requires $v(3) = 0$, or $v(3) = 1$, which requires $v(2) = 0$.  Either way, we get an upper bound of $2$ for the sum.
\end{example}

\begin{proposition}
    For any contextuality scenario $\mathfrak{C}$, $C_\mathfrak{C} \subseteq Q_{\mathfrak{C}} \subseteq S_{\mathfrak{C}}$.  There exist contextuality scenarios in which both inclusions are strict.
\end{proposition}
\begin{proof}
    The inclusion of $C_{\mathfrak{C}}$ and $Q_{\mathfrak{C}}$ in $S_{\mathfrak{C}}$ is trivial, since both are defined as subsets of states, so we only have to prove $C_\mathfrak{C} \subseteq Q_{\mathfrak{C}}$.  \Cref{prop:classinquant} shows how to construct a quantum model from the set of value functions.  If we have a KS noncontextual state of the form $\omega(a) = \sum_{v\in V_{\mathfrak{C}}}p_v v(a)$ then we can construct a density operator $\rho = \sum_{v\in V_{\mathfrak{C}}} p_v \ket{v}\bra{v}$ on the Hilbert space of the corresponding model.  It is straightforward to show that this yields the same probabilities.

    For the strictness, consider a noncontextuality inequality $\sum_{a\in X} c_a \omega(a) \leq \gamma_c$ and let $\gamma_q$ be the largest value of $\sum_{a\in X} c_a \omega(a)$ obtainable from a quantum state.  If $\gamma_q > \gamma_c$ and there exists a state with $\sum_{a\in X} c_a \omega(a) > \gamma_q$ then the inclusions are strict.  The Klyachko scenario and inequality are an example of this.  It can be shown that $\gamma_q = \sqrt{5} > 2 = \gamma_c$ for this scenario \cite{Cabello_NonContextualityPhysical_2010, Liang_2011, Cabello_GraphTheoreticApproachQuantum_2014}.  However, $\omega(0) = \omega(1) = \omega(2) = \omega(3) = \omega(4) = 1/2$ is a valid state and this has $\sum_{a\in X}\omega(a) = 5/2 > \sqrt{5}$. 
\end{proof}

\section{Antidistinguishability}
\label{Anti}

In this section, we review the concept of \emph{antidistinguishability}, which was originally introduced under the name \emph{PP-incompatibility} in \cite{Caves_ConditionsCompatibilityQuantumstate_2002} and re-branded as antidistinguishability in \cite{Leifer_QuantumStateReal_2014}.  Although antidistinguishability is usually discussed for sets of quantum states, here we define it for sets of outcomes in a contextuality scenario.  In a quantum contextuality scenario, the outcomes, which are elements of orthonormal bases, can also be regarded as pure quantum states.  Therefore, in a quantum contextuality scenario, antidistinguishability of outcomes and of pure quantum states amounts to the same thing.  In a general contextuality scenario, where there need not be a self-duality between states and measurement outcomes, this would not be the case.  Although the concept of antidistinguishability of states is more natural, antidistinguishability of outcomes is what we need to prove noncontextuality inequalities.

We start this section by giving our general definition, then explain how it reduces to the usual definition for quantum contextuality scenarios, and then state a useful theorem from \cite{Caves_ConditionsCompatibilityQuantumstate_2002} that characterizes antidistinguishability for sets of three pure quantum states.  We will use this to establish examples of antidistinguishability-based noncontextuality inequalities in \S\ref{Exa}.

\begin{definition}
    \label{def:Anti}
    In a contextuality scenario $\mathfrak{C} = (X,\mathcal{M},\mathcal{N})$, a set of outcomes $\{a_1,a_2,\cdots,a_n\} \subseteq X$ is \emph{antidistinguishable} if there exists outcomes $a_1^{\perp},a_2^{\perp},\cdots,a_n^{\perp} \in X$ such that 
    \begin{itemize}
        \item There exists a context $M \in \mathcal{M}$ with $\{a_1^{\perp},a_2^{\perp},\cdots,a_n^{\perp}\} \subseteq M$. 
        \item For each $j \in [n]$, there exists a context or a maximal partial context $N_j$ such that $\{a_j,a_j^{\perp}\} \subseteq N_j$.
        \item For each outcome $a \in M \backslash \{a_1^{\perp},a_2^{\perp},\cdots,a_n^{\perp}\}$ and each $a_j$, there exists a context or maximal partial context $N$ such that $\{a,a_j\}\subseteq N$.
    \end{itemize}
\end{definition}

\Cref{exa:anti} is a simple example of a set of three antidistinguishable outcomes.

To understand this better, it is useful to look at how \cref{def:Anti} applies to the quantum case in more detail. 

\begin{example}
A set $\{\ket{a_1},\cdots,\ket{a_n}\}$, $n \leq d$ of states in $\mathbb{C}^d$ is antidistinguishable if there exists an orthonormal basis $\{ \ket{a_{1}^{\perp}},\cdots,\ket{a_{n}^{\perp}},\cdots,\ket{a_d^{\perp}} \}$ such that
\begin{equation}
\braket{a_j^{\perp} | a_{j}} = 0, \,\, \forall \,\, j \in [n]
\label{Eq:DefAnti1}
\end{equation}
and
\begin{equation}
 \braket{a_k^{\perp} | a_{j}} = 0 , \,\, \forall \,\, j \in [n], k\in [n+1 ,d].
\label{Eq:DefAnti2}
\end{equation}
\label{Def:AntiDist}
\end{example}

The idea of antidistinguishability for states is that if one of the states $\ket{a_1},\cdots,\ket{a_n}$ is prepared and you do not know which then there exists a measurement that allows you to definitively rule out one of the states.  It should be contrasted with distinguishability in which there exists a measurement that allows you to tell exactly which state was prepared.  Antidistinguishability is weaker than distinguishability.

\Cref{Eq:DefAnti2} states that the vectors $\ket{a_j}$ are in the subspace spanned by $\ket{a_k^{\perp}}$ for $k\in [n]$.  This rules out the trivial case where we choose all these $\ket{a_k^{\perp}}$ to be orthogonal to every $\ket{a_j}$ for every $j$.  This is also the reason for the third clause in \cref{def:Anti}.

The following theorem from \cite{Caves_ConditionsCompatibilityQuantumstate_2002}, provides a useful characterization of antidistinguishability for sets of three pure states, as it avoids the need to construct the antidistinguishing measurement explicitly.
\begin{thm}
    Consider a set $\mathcal{A} = \{\ket{a_1},\ket{a_2},\ket{a_3}\}$ of three states and let $x_1 = \vert \braket{a_2 | a_3} \vert^2$, $x_2 = \vert \braket{a_1 | a_3} \vert^2$, $x_3 = \vert \braket{a_1 | a_2} \vert^2$.  Then, $\mathcal{A}$ is antidistinguishable iff
    \begin{align}
        x_1 + x_2 + x_3 & < 1 \label{Eq:Anti1} \\
        (x_1 + x_2 + x_3 - 1)^2 & \geq 4 x_1x_2x_3. \label{Eq:Anti2}
    \end{align}
    \label{Thm:Anti3}
\end{thm}
The following corollary, as stated in \cite{Havlicek_2019}, gives a simpler sufficient condition for antidistinguishability that is easier to check.  It follows by substitution into \cref{Eq:Anti1} and \cref{Eq:Anti2}.
\begin{corollary}
    \label{Cor:Anti}
Consider a set $\mathcal{A} = \{\ket{a_1},\ket{a_2},\ket{a_3}\}$ of three states and let $x_1 = \vert \braket{a_2 | a_3} \vert^2$, $x_2 = \vert \braket{a_1 | a_3} \vert^2$, $x_3 = \vert \braket{a_1 | a_2} \vert^2$.  Then, $\mathcal{A}$ is antidistinguishable if
\begin{equation}
    x_1,x_2,x_3 \leq \frac{1}{4}.
\end{equation}
\end{corollary}
Additional criteria for antidistinguishability have been proved for more general cases \cite{Bandyopadhyay_ConclusiveExclusionQuantum_2014,Heinosaari_AntidistinguishabilityPureQuantum_2018}, but we shall not need them here.

\section{Noncontextuality Inequalities from Antidistinguishability}
\label{Ineq}

This section describes our main results.  We can use the concept of antidistinguishability to derive noncontextuality inequalities based on \emph{pairwise antisets}.  These come in two versions---strong and weak---which are used to derive state independent and state dependent inequalities respectively.

The notion of a weak pairwise antiset, applied to states rather than outcomes and not explicitly named, was used in \cite{Barrett_2014} to derive overlap bounds on the reality of the quantum state.  Other examples of this construction were given in \cite{Branciard_HowPsepistemicModels_2014, Knee_2017}.  In light of our results, these bounds can now be reinterpreted as state dependent noncontextuality inequalities.  The notion of a strong pairwise antiset is novel to this work, and allows us to show that some of these inequalities are actually state independent.

After defining pairwise antisets and stating our main results, \S\ref{Exa} gives examples of our construction for quantum contextuality scenarios.  The proof of our main results is given in \S\ref{Proof}.

\begin{definition}\label{Def:StrongPairwiseAS}
    A \emph{strong pairwise antiset} $W$ in a contextuality scenario $\mathfrak{C} = (X,\mathcal{M},\mathcal{N})$ is a set of outcomes for which there exists a context $M \in \mathcal{M}$ such that, for every $a,b \in W$ and $c \in  M$, the triple $\{a,b,c\}$ is antidistinguishable.
    
    The context $M$ is called a \emph{principal context} for the pairwise antiset $W$.
\end{definition}

\begin{definition}
    A \emph{weak pairwise antiset} $W$ in a contextuality scenario $\mathfrak{C} = (X,\mathcal{M},\mathcal{N})$ is a set of outcomes for which there exists another outcome $c \in X$ such that, for every $a,b \in W$, the triple $\{a,b,c\}$ is antidistinguishable.
    
    The outcome $c$ is called a \emph{principal outcome} for the pairwise antiset $W$.
\end{definition}

We are now in a position to state our main results.

\begin{restatable}{thm}{mainresult}
\label{Thm:MainResult}
    Let $W$ be a pairwise antiset in a contextuality scenario $\mathfrak{C} = (X,\mathcal{M},\mathcal{N})$.  If $W$ is strong then any state $\omega \in C_{\mathfrak{C}}$ satisfies
\begin{equation}
    \label{eq:mainresult}
    \sum_{a \in W} \omega(a) \leq 1.
\end{equation}
    If $W$ is weak then any $\omega \in C_{\mathfrak{C}}$ that also satisfies $\omega(c) = 1$ for a principal outcome $c$ satisfies \cref{eq:mainresult}.
\end{restatable}

\section{Examples}
\label{Exa}

Before proving \cref{Thm:MainResult}, here are some interesting examples of pairwise antisets that occur in quantum contextuality scenarios and the noncontextuality inequalities that arise from them.

\subsection{Strong Pairwise Antisets}

In this section, we give examples of strong pairwise antisets and state independent inequalities.

\begin{example}[The Yu-Oh inequality]
As a first example, we re-derive a noncontextuality inequality first given in \cite{Yu_StateIndependentProofKochenSpecker_2012} using \cref{Thm:MainResult}.
Consider the following four vectors in $\mathbb{C}^3$
\begin{align}
\ket{a_1}= \frac{1}{\sqrt{3}}\begin{pmatrix}
           1 \\
           1 \\
           1
         \end{pmatrix}, \,\,
\ket{a_2}= \frac{1}{\sqrt{3}}\begin{pmatrix}
           -1 \\
           1 \\
           1
         \end{pmatrix}, \nonumber \\
\ket{a_3}= \frac{1}{\sqrt{3}}\begin{pmatrix}
           1 \\
           -1 \\
           1
         \end{pmatrix}, \,\,
\ket{a_4}= \frac{1}{\sqrt{3}}\begin{pmatrix}
           1 \\
           1 \\
           -1
         \end{pmatrix}.
\label{Eq:YuOhRays}
\end{align}
These form a strong pairwise antiset with principal basis
\begin{align}
\ket{c_1}= \begin{pmatrix}
           1 \\
           0 \\
           0
         \end{pmatrix}, \,\,
\ket{c_2}= \begin{pmatrix}
           0 \\
           1 \\
           0
         \end{pmatrix}, \,\,
\ket{c_3}= \begin{pmatrix}
           0 \\
           0 \\
           1
         \end{pmatrix}.
\label{Eq:YuOhBasis}
\end{align}

The triple $\{\ket{c_1},\ket{a_1},\ket{a_2}\}$ was shown to be antidistinguishable in \cref{exa:qanti}.  The other triples $\{\ket{c_j},\ket{a_k},\ket{a_m}\}$ for $k\neq m$ are antidistinguishable because they have the same inner products so they satisfy the conditions of \cref{Thm:Anti3}.  \Cref{Thm:MainResult} thus implies the noncontextuality inequality
\begin{equation}
    \label{Eq:YuOhIneq}
    \sum_{j=1}^4 \omega(a_j) \leq 1.
\end{equation}
However, the four states $\ket{a_j}$ satisfy
\[\sum_{j=1}^4 \ket{a_j}\bra{a_j} = \frac{4}{3} I,\]
where $I$ is the identity operator.  This implies that for any quantum state $\omega$, the quantum predictions are
\begin{equation}
    \sum_{j=1}^4 \omega(a_j) = \frac{4}{3} > 1.
\end{equation}
\end{example}

In \cite{Yu_StateIndependentProofKochenSpecker_2012}, the inequality of \cref{Eq:YuOhIneq} was derived by applying an exhaustive search over noncontextual assignments to the orthogonality graph of $13$ rays in $\mathbb{C}^3$ \footnote{See \cite{Cabello_ProposedExperimentsQutrit_2012} for more details.}.  Here we only used $7$ rays, but the other rays used in \cite{Yu_StateIndependentProofKochenSpecker_2012} are just the elements of the orthonormal bases that are required to antidistinguish the triples used in our argument.  Re-deriving the inequality using \cref{Thm:MainResult} shows that it was based on antidistinguishability all along, and this allows us to easily generalize the example. 

\begin{example}[Hadamard States]
The Yu-Oh construction can be generalized as follows.  Consider the following vectors in $\mathbb{C}^d$:
\begin{equation}
    \ket{a_{\bm{x}}}= \frac{1}{\sqrt{d}}\begin{pmatrix}
       -1^{x_1} \\
       -1^{x_2} \\
       -1^{x_3} \\
       \vdots \\
       -1^{x_d}
    \end{pmatrix},
\end{equation}
where $\bm{x}=(x_1,\cdots,x_d)$ is a binary vector in $\{0,1\}^{d}$.  This means that, ignoring normalization for the moment, the components of $\ket{a_{\bm{x}}}$ are all either $+1$ or $-1$ and as we run through the possible vectors $\bm{x}$ we get all possible combinations of $\pm 1$ components.  There are $2^d$ such vectors.  These vectors are called \emph{Hadamard states} because they can be thought of as the possible columns of Hadamard matrices.  In addition, let $\{\ket{0},\ket{1},\cdots,\ket{d-1}\}$ be the standard orthonormal basis for $\mathbb{C}^d$, which we will use as the principal basis (in the sense of def. \ref{Def:StrongPairwiseAS}). 

Now, obviously, not all triples $\{\ket{j},\ket{a_{\bm{x}}},\ket{a_{\bm{x}'}}\}$ are antidistinguishable because some pairs $\ket{a_{\bm{x}}},\ket{a_{\bm{x'}}}$ only differ by a phase, i.e.\ $\ket{a_{\bm{x}'}} = -\ket{a_{\bm{x}}}$.  In this case, $ \left | \braket{a_{\bm{x}} | a_{\bm{x}'}} \right |^2 = 1$ and so \cref{Eq:Anti1} of \cref{Thm:Anti3} is not satisfied.  We can eliminate such cases by only considering binary vectors $\bm{x}$ that begin with a $0$.  Denote this set $B^d_0$ and the set of binary strings that begin with a $1$ by $B^d_1$.  Both sets contain $2^{d-1}$ vectors.

Restricting to $B^d_0$, the triples $\{\ket{j},\ket{a_{\bm{x}}},\ket{a_{\bm{x'}}}\}$ satisfy the conditions of \cref{Thm:Anti3} for $\bm{x} \neq \bm{x}'$ and so \cref{Thm:MainResult} implies that noncontextual states satisfy
\begin{equation}
    \label{Eq:HadamardIneq1}
    \sum_{\bm{x} \in B_0^d} \omega(a_{\bm{x}}) \leq 1.
\end{equation}

Since the vectors in $B_1^d$ represent the same set of rays, we can run the same argument and obtain
\begin{equation}
    \label{Eq:HadamardIneq2}
    \sum_{\bm{x} \in B_1^d} \omega(a_{\bm{x}}) \leq 1.
\end{equation}
Adding the two inequalities gives
\begin{equation}
    \label{Eq:HadamardIneq3}
    \sum_{\bm{x} \in \{0,1\}^d} \omega(a_{\bm{x}}) \leq 2.
\end{equation}
Although it is not necessary to add the inequalities like this, it is a bit cleaner to work with the full set of vectors of size $2^d$ rather than two sets of size $2^{d-1}$.

For the quantum probabilities we note that 
\[\left ( \sum_{\bm{x} \in \{0,1\}^d} \ket{a_{\bm{x}}}\bra{a_{\bm{x}}} \right )_{jk} = \frac{1}{d} \sum_{\bm{x} \in \{0,1\}^d} (-1)^{x_j + x_k}.\]
For $j=k$, each term in the sum is $+1$, so the diagonal components are all $2^d/d$.  For $j\neq k$, the off-diagonal components are all $0$ because there are as many vectors in which $x_j = x_k$ as there are in which $x_j \neq x_k$ so there are an equal number of $+1$'s and $-1$'s in the sum.

Thus, we have
\[\sum_{\bm{x} \in \{0,1\}^d} \ket{a_{\bm{x}}}\bra{a_{\bm{x}}} = \frac{2^d}{d} I,\]
so the probabilities for any quantum state $\omega$ are
\begin{equation}
    \sum_{\bm{x}\in\{0,1\}^d} \omega(a_{\bm{x}}) = \frac{2^d}{d},
\end{equation}
This is larger than $2$ whenever $d \geq 3$, which yields another state independent contextuality proof.
\end{example}

Hadamard states, combined with the Frankl-R\"{o}dl theorem \cite{Frankl_1987}, have previously been used to prove noncontextuality inequalities and to bound quantum information protocols \cite{Buhrman_1998, Brassard_1999, Mancinska_2013}.   From a modern perspective, this amounts to considering the orthogonality properties of Hadamard states instead of their antidistinguishability, and applying the CSW formalism \cite{Cabello_NonContextualityPhysical_2010,Cabello_GraphTheoreticApproachQuantum_2014}.  From this, we find that there exists an $\epsilon > 0$ such that
\begin{equation}
    \label{Eq:HadamardIneq4}
    \sum_{\bm{x} \in \{0,1\}^d} \omega(a_{\bm{x}}) \leq (2 - \epsilon)^d,
\end{equation}
for every $\omega \in C_{\mathfrak{C}}$.
While this also proves contextuality for sufficiently large $d$, the bound is a lot larger than that of \cref{Eq:HadamardIneq3}, which shows the benefit of considering antidistinguishability.

In \cite{Leifer_PsEpistemicModelsAre_2014}, one of the authors of the present paper used the noncontextuality inequality of \cref{Eq:HadamardIneq4} to derive an overlap bound constraining $\psi$-epistemic models.  It was subsequently pointed out by Maroney \cite{Maroney_2014} and Branciard \cite{Branciard_HowPsepistemicModels_2014} that the overlap bound could be tightened along the lines of \cref{Eq:HadamardIneq3} using antidistinguishability.  The innovation here is to recognize that \cref{Eq:HadamardIneq3} is also a noncontextuality inequality.

The next example was also first proposed as an overlap bound in \cite{Barrett_2014}, which we can now recognize as a noncontextuality inequality.  

\begin{example}[Mutually Unbiased Basis (MUBs)]
Two orthonormal bases $\{\ket{e_j}\}_{j=1}^d$ and $\{\ket{f_j}\}_{j=1}^d$ in $\mathbb{C}^d$ are \emph{mutually unbiased} if $\left | \braket{e_j | f_k} \right |^2 = 1/d$ for all $j$ and $k$.  When $d$ is a prime power, then $d+1$ mutually unbiased bases are known to exist \cite{Bengtsson_2007}.  Let $\{\ket{a_{jk}}\}$ be the set of all vectors that appear in one of these basis, where $j$ runs over the choice of basis from $1$ to $d+1$ and $k$ runs over the vectors within a basis from $1$ to $d$.  We remove one basis, say $ \{\ket{a_{1k}}\}_{k=1}^d$, to be our principal basis, so there are $d^2$ vectors left in the set.

We have $\left | \braket{a_{jk}|a_{j'k'}} \right |^2 = \delta_{jj'}\delta_{kk'} + (1-\delta_{jj'})\frac{1}{d}$ and, for $d\geq 4$, \cref{Cor:Anti} implies that $\{\ket{a_{1k}},\ket{a_{j'k'}},\ket{a_{j''k''}}\}$ is antidistinguishable whenever $j'k'$ is distinct from $j''k''$ and $j',j'' \neq 1$.  Thus, we have a strong pairwise antiset so \cref{Thm:MainResult} implies that
\begin{equation}
     \left [ \sum_{j=1}^{d+1} \sum_{k=1}^d \omega(a_{jk}) \right ] - \sum_{k=1}^d \omega(a_{1k}) \leq 1,
\end{equation}
for any state $\omega \in C_{\mathfrak{C}}$. In fact, since $\{a_{1k}\}_{k=1}^d$ is a context, we have $\sum_{k=1}^d \omega(a_{1k}) = 1$ for any state $\omega$, so we have
\begin{equation}
    \label{Eq:MUBIneq}
     \sum_{j=1}^{d+1} \sum_{k=1}^d \omega(a_{jk}) \leq 2.
\end{equation}

Since $\{\ket{a_{jk}}\}_{k=1}^d$ is an orthonormal basis, we have
\[\sum_{j=1}^{d+1} \sum_{k=1}^d \ket{a_{jk}}\bra{a_{jk}} = (d+1) I,\]
so the quantum probabilities are
\begin{equation}
    \sum_{j=1}^{d+1} \sum_{k=1}^d \omega(a_{jk}) = d + 1,
\end{equation}
for any quantum state $\omega \in Q_{\mathfrak{C}}$.
This violates \cref{Eq:MUBIneq} for $d \geq 3$, but recall that the antidistinguishability conditions only hold for $d \geq 4$, so this is a contextuality proof for prime power $d \geq 4$.
\end{example}

\subsection{Weak Pairwise Antisets}

In this section, we give examples of state dependent noncontextuality inequalities arising from weak pairwise antisets.

The following simple example is due to Owen Maroney \cite{Maroney_2014}.

\begin{example}[Maroney States]
Consider the following vectors in $\mathbb{C}^d$
\begin{equation}
    \ket{a_j} = \frac{1}{\sqrt{3}} \ket{0} + \sqrt{\frac{2}{3}} \ket{j},
\end{equation}
where $j$ runs from $1$ to $d-1$ and we denote the standard orthonormal basis vectors as $\ket{0},\ket{1},\cdots,\ket{d-1}$.  We also set $\ket{c} = \ket{0}$.

Using \cref{Thm:Anti3}, we can easily check that $\{\ket{c},\ket{a_j},\ket{a_k}\}$ is antidistinguishable for $j \neq k$, so we have a weak pairwise antiset $W = \{\ket{a_j}\}_{j=1}^{d-1}$ and principal outcome $\ket{c}$. \Cref{Thm:MainResult} then gives
\begin{equation}
    \sum_{j=1}^{d-1} \omega(a_j) \leq 1,
\end{equation}
for any noncontextual state $\omega$ such that $\omega(c)=1$.

The quantum state $\omega$ corresponding to the vector $\ket{c} = \ket{0}$ obviously satisfies $\omega (c) = 1$ and it has $\omega(a_j) = \left | \braket{a_j|c}\right |^2 = 1/3$ for all $j$ so we get
\begin{equation}
    \sum_{j=1}^{d-1} \omega(a_j) = \frac{d-1}{3}.
\end{equation}
This proves that $\omega$ is contextual in this scenario for $d \geq 5$.
\end{example}

\begin{example}[Symmetric Informationally Complete (SIC) POVMs]
A SICPOVM, or SIC for short, is a set of semi-positive operators $\{E_j\}_{j=1}^{d^2}$ on $\mathbb{C}^d$ that satisfy
\begin{equation}
    \label{eq:SICnorm}
    \sum_{j=1}^{d^2} E_j = I,
\end{equation}
and are of the form $E_j = \frac{1}{d} \ket{a_j}\bra{a_j}$ where
\begin{equation}
    \label{eq:SICdef}
    \left | \braket{a_j|a_k} \right |^2 = \frac{1}{d+1},
\end{equation}
for $j \neq k$.  SICs are conjectured to exist in all finite Hilbert space dimensions.  They have been shown to exist in all dimensions up to $d=151$ and in several larger dimensions up to $d=844$ \cite{Fuchs_2017}.

For a SIC, let $\ket{c} = \ket{a_1}$ and $W = \{\ket{a_j}\}_{j=2}^{d^2}$.  \Cref{Cor:Anti} implies that, for $d \geq 3$, the triples $\{\ket{c},\ket{a_j},\ket{a_k}\}$ are all antidistinguishable for $j\neq k$ and $j,k \neq 1$ so we have a weak pairwise antiset.  Thus, \cref{Thm:MainResult} implies that
\begin{equation}
     \left [ \sum_{j=1}^{d^2} \omega(a_j) \right ] - \omega(a_1) \leq 1,
\end{equation}
for any noncontextual state $\omega$ such that $\omega(c)=1$.

Since $c = a_1$, we obviously also have $\omega(a_1) = 1$, so
\begin{equation}
    \label{eq:SICInequality}
     \sum_{j=1}^{d^2} \omega(a_j) \leq 2,
\end{equation}
for any noncontextual state $\omega$ such that $\omega(c)=1$.

Now consider any quantum state $\omega$.  From \cref{eq:SICnorm}, we have
\[\sum_{j=1}^{d^2} \ket{a_j}\bra{a_j} = d I,\]
so the quantum predictions are
\begin{equation}
    \sum_{j=1}^{d^2} \omega(a_j) = d,
\end{equation}
If we also have $\omega(c)=1$, which is the case for the quantum state corresponding to $\ket{a_1}$ for example, then this state is contextual for $d \geq 3$.
\end{example}

For $d=3$, the inequality of \cref{eq:SICInequality} was derived as a \emph{state independent} contextuality inequality in \cite{Bengtsson_2012} based on a special relationship between MUBs and SICs that only occurs in that dimension.  They considered the orthogonality graph of $21$ vectors in $\mathbb{C}^3$ consisting of the vectors that appear in a SIC and those that appear in a related set of $4$ MUBs.  From our perspective, the special relationship is that, in $d=3$, MUBs can be chosen that antidistinguish each of the triples $\{\ket{c},\ket{a_j},\ket{a_k}\}$ used in our proof.  

Our generalization follows from the fact that these antidistinguishability relations still hold in higher dimensions, but the antidistinguishing measurements are no longer necessarily MUBs.  Unfortunately, our generalization is only a state dependent inequality, as we did not find a way of generating a principal context from a SIC\@.  This indicates that other methods of generating noncontextuality inequalities from antidistinguishability might exist.

\section{Proof of Theorem~\ref{Thm:MainResult}}

\label{Proof}

\mainresult*

The proof is based on one lemma and the Bonferroni inequalities \cite{Rohatgi_2001}.  

\begin{lem}
Let $A$ be a set of antidistinguishable outcomes in a contextuality scenario $\mathfrak{C} = (X,\mathcal{M},\mathcal{N})$.  Then, there are no value functions that are $a$-definite for every $a\in A$, i.e.
\begin{equation}
    \bigcap_{a \in A} V_a = \emptyset.
\end{equation}
\label{Lem:AntiImpliesEmpty}
\end{lem}

\begin{proof}
    Suppose $A = \{a_1,a_2,\cdots,a_n\}$ has $n$ outcomes and that $M = \{a_1^{\perp},a_2^{\perp},\cdots,a_m^{\perp}\}$ (with $m \geq n$) is a context that antidistinguishes them, i.e.
    \begin{itemize}
        \item For all $j \in [n]$, there exists a context or maximal partial context $N_j$ such that $a_j,a_j^{\perp} \in N_j$.
        \item For all $j \in [n]$ and $k \in [n+1,m]$, there exists a context or maximal partial context $N_{jk}$ such that $a_j,a_k^{\perp} \in N_{jk}$.
    \end{itemize}
    Since $\{a_1^{\perp},a_2^{\perp},\cdots,a_m^{\perp}\}$ is a context, every value function must assign the value $1$ to exactly one outcome in this set.  Consider a value function $v \in V_{a_1}$.  Since, for every $k \in [n+1,m]$, there is a (maximal partial) context that contains both $a_1$ and $a_k^{\perp}$, $v$ must assign value $0$ to every $a_k^{\perp}$ for $k \in [n+1,m]$.  This means that it must assign value $1$ to one of $a_1^{\perp}, a_2^{\perp},\cdots,a_n^{\perp}$.
    
    Now suppose that $v \in \bigcap_{a \in A} V_a$ so that it assigns value $1$ to every $a_j$ for $j \in [n]$.  It cannot assign $v(a_j^{\perp}) = 1$ for any $j = [n]$ because, for every such $j$, there is always a (maximal partial) context $N_j$ such that $a_j,a_j^{\perp} \in N_j$ and we already have $v(a_j) = 1$.  This means that the value function must assign value $0$ to every $a_j^{\perp}$ for $j \in [m]$, contradicting the requirement that it assign value $1$ to exactly one outcome in every context.  Therefore, no such value function exists, so $\bigcap_{a \in A} V_a = \emptyset$.
\end{proof}

\begin{proof}[Proof of \Cref{Thm:MainResult}]
Given a contextuality scenario $\mathfrak{C} = (X,\mathcal{M},\mathcal{N})$ and a noncontextual state $\omega \in C_{\mathfrak{C}}$, which is necessarily of the form
\[\omega(a) \sum_{v\in V_{\mathfrak{C}}} p_v v(a),\]
we can define a probability space $(V_{\mathfrak{C}},2^{V_{\mathfrak{C}}},P)$ over the value functions via
\[P(V) = \sum_{v\in V_{\mathfrak{C}}} p_v.\]

Now consider the quantity
\begin{equation}
    \sum_{a \in W} \omega(a) = \sum_{a \in W} \sum_{v\in V_{\mathfrak{C}}} p_v v(a).
\end{equation}
Because $v(a) = 1$ iff $v \in V_a$ and $v(a) = 0$ otherwise, we can rewrite this as
\begin{equation}
    \label{eq:OmegaP}
    \sum_{a \in W} \omega(a) = \sum_{a \in W}  P(V_a).
\end{equation}

Next, we make use of the Bonferroni inequalities \cite{Rohatgi_2001}.  Recall that the Bonferroni inequalities are a generalization of the inclusion-exclusion principle to probability spaces.  For a probability space $(\Omega,\Sigma,P)$, let $\Omega_1,\Omega_2,\cdots,\Omega_n \in \Sigma$ be measurable sets.  Then, we have the sequence of inequalities:
\begin{align}
    P\left ( \bigcup_{j=1}^n \Omega_j \right ) & \leq \sum_{j = 1}^n P \left ( \Omega_j \right ) \\
    P\left ( \bigcup_{j=1}^n \Omega_j \right ) & \geq \sum_{j = 1}^n P \left ( \Omega_j \right ) - \sum_{j < k} P \left ( \Omega_j \cap \Omega_k \right ) \label{eq:Bonf} \\
    P\left ( \bigcup_{j=1}^n \Omega_j \right ) & \leq \sum_{j = 1}^n P \left ( \Omega_j \right ) - \sum_{j < k} P \left ( \Omega_j \cap \Omega_k \right ) \nonumber \\ & + \sum_{j < k < l} P \left ( \Omega_j \cap \Omega_k \cap \Omega_l \right ) \\
    \vdots & \nonumber
\end{align}
The pattern continues with alternating signs of the additional terms and alternating directions of the inequalities.  Here, we will make use of the second Bonferroni inequality given in \cref{eq:Bonf}.

Suppose that the pairwise antiset $W$ has $n$ outcomes $W = \{a_1,a_2,\ldots,a_n\}$ and consider the corresponding sets of $a_j$-definite value functions $V_{a_1},V_{a_2},\cdots,V_{a_n}$.  By the second Bonferroni inequality we have
\begin{equation}
     P\left ( \bigcup_{j=1}^n V_{a_j} \right ) \geq \sum_{j = 1}^n P \left ( V_{a_j} \right ) - \sum_{j < k} P \left ( V_{a_j} \cap V_{a_k} \right ).
\end{equation}
Combining this with \cref{eq:OmegaP} and rearranging gives
\begin{equation}
    \sum_{j=1}^n \omega(a_j) \leq P\left ( \bigcup_{j=1}^n V_{a_j} \right ) + \sum_{j < k} P \left ( V_{a_j} \cap V_{a_k} \right ).
\end{equation}
Because $P$ is a probability measure, we have $P\left ( \bigcup_{j=1}^n V_{a_j} \right ) \leq 1$, so 
\begin{equation}
    \sum_{j=1}^n \omega(a_j) \leq 1 + \sum_{j < k} P \left ( V_{a_j} \cap V_{a_k} \right ).
\end{equation}
Therefore, the theorem follows if we can show that $P \left ( V_{a_j} \cap V_{a_k} \right ) = 0$ for every $j\neq k$.  To do this, we consider the cases of strong and weak pairwise antisets separately.

In the strong case, consider a principal context $M$.  Since it is a context the sets $V_c$ for $c \in M$ are disjoint and form a partition of $V_{\mathfrak{C}}$, so $V_{\mathfrak{C}} = \bigcup_{c \in M} V_c$.  Also, $P(V_{\mathfrak{C}}) = 1$, so we have
\begin{align}
    P \left ( V_{a_j} \cap V_{a_k} \right ) & = P \left ( V_{\mathfrak{C}} \cap \left [ V_{a_j} \cap V_{a_k} \right ] \right )\\
    & = P \left ( \left [ \bigcup_{c \in M} V_c \right ] \cap \left [ V_{a_j} \cap V_{a_k} \right ] \right ) \\
    & = \sum_{c \in M} P \left (  V_c \cap V_{a_j} \cap V_{a_k}  \right ),
\end{align}
where the last line follows from disjointness of the $V_c$'s.  Since each triple $\{c,a_j,a_k\}$ is antidistinguishable, \cref{Lem:AntiImpliesEmpty} implies that $V_c \cap V_{a_j} \cap V_{a_k} = \emptyset$, and hence $P(V_c \cap V_{a_j} \cap V_{a_k}) = 0$.

Now consider the weak case.  Let $c$ be a principal outcome and suppose that $\omega(c) = 1$.  Since $\omega(c) = P(V_c)$ it follows that $P(V_c) = 1$.  Thus, we can write
\begin{align}
    P \left ( V_{a_j} \cap V_{a_k} \right ) = P \left ( V_c \cap V_{a_j} \cap V_{a_k} \right ),
\end{align}
and since each triple $\{c,a_j,a_k\}$ is antidistinguishable, \cref{Lem:AntiImpliesEmpty} implies that $V_c \cap V_{a_j} \cap V_{a_k} = \emptyset$, and hence $P(V_c \cap V_{a_j} \cap V_{a_k}) = 0$.
\end{proof}

\section{Conclusions}

\label{Conc}

In this paper, we have shown that the antidistinguishability properties of sets of quantum states, and more abstractly outcomes in a contextuality scenario, can be used to derive noncontextuality inequalities.  Our method can be used to re-derive some known inequalities, such as the Yu-Oh inequality \cite{Yu_StateIndependentProofKochenSpecker_2012}, in a simple way that uncovers the previously hidden antidistinguishability structure of the proof.  It can also be used to generalize known inequalities to higher dimensions, such as in the Hadamard and SIC examples, and derive new classes of noncontextuality inequalities, such as the example based on MUBs.  In some cases, we get much tighter bounds on the inequalities than we would get from considering the distinguishability properties alone, such as in the Hadamard example.  Our method is not necessarily the only way of deriving noncontextuality inequalities from antidistinguishability, and we think there is much to be gained from considering antidistinguishability structures further, particularly given their role in some recently proposed quantum information protocols \cite{Perry_CommunicationTasksInfinite_2015, Havlicek_2019}.

In principle, our noncontextuality inequalities could be made robust to noise and tested experimentally using the techniques described in \cite{Kunjwal_2018, ADO18}.  However, in order to do so, one would have to experimentally test that the antidistinguishabilities used in the proofs hold approximately in the lab.  This would involve constructing the bases that antidistinguish the states in our pairwise antisets, increasing the number of vectors needed to establish the proof.  It would then essentially reduce to a proof based on the orthogonality properties of the states.  From a theoretical point of view, one of the virtues of our method is that you do not have to explicitly construct the antidistinguishing measurements, so we can derive our inequalities using a smaller number of vectors than would be needed in methods based on orthogonality.  This advantage would be lost in the experimental tests.

Thus, we think the main use of our method will be in theoretical work, where contextuality inequalities can be used to prove things about quantum computation and quantum information protocols.  As an example of this, the amount of memory needed to classically simulate stabilizer quantum computations was recently bounded using contextuality proofs based on antidistinguishability \cite{Karanjai_2018}.  We expect that having a general method of constructing inequalities based on antidistinguishability could be used to prove similar and more general results for other classes of quantum computation.

Our work also has implications for thinking about overlap bounds on the reality of the quantum state.  One known class of bounds is based on CSW noncontextuality inequalities, but the other class---based on antidistinguishability---did not previously have a known connection to contextuality.  In this paper, we have shown that this second class of bounds are also noncontextuality inequalities.  It has been shown that a maximally $\psi$-epistemic model (one in which the quantum and classical overlaps are equal) must be noncontextual \cite{Leifer_MaximallyEpistemicInterpretations_2013, Leifer_QuantumStateReal_2014}, which explains why contextuality proofs provide overlap bounds.  However, the converse is not necessarily true.  This indicates that better overlap bounds than those currently known might be obtainable by considering the constraints on maximally $\psi$-epistemic models that are not implied by noncontextuality.

\begin{acknowledgments}
Matthew Leifer wishes to thank Owen Maroney and Matt Pusey for useful discussions.

This research was supported in part by the Fetzer Franklin Fund of the John E.\ Fetzer Memorial Trust and by grant number FQXi-RFP-IPW-1905 from the Foundational Questions Institute and Fetzer Franklin Fund, a donor advised fund of Silicon Alley Community Foundation.  

Matthew Leifer is grateful for the hospitality of Perimeter Institute where part of this work was carried out. Research at Perimeter Institute is supported in part by the Government of Canada through the Department of Innovation, Science and Economic Development Canada and by the Province of Ontario through the Ministry of Economic Development, Job Creation and Trade.

Cristhiano Duarte was supported by a fellowship from the Grand Challenges Initiative at Chapman University. 
\end{acknowledgments}

\bibliography{list_of_references}
\end{document}